\documentclass{iopart}
\usepackage{amsthm}
\usepackage{verbatim}
\usepackage{graphicx,color} 
\usepackage{booktabs} 
\usepackage{color}

\graphicspath{{./figures/}}
\usepackage{hyperref}

\hypersetup{
  pdftitle={Computational complexity of time-dependent density functional theory},
  colorlinks,%
  citecolor=black,%
  filecolor=black,%
  linkcolor=black,%
  urlcolor=black,
  pdfauthor={J. D. Whitfield et al.},
}
\newcommand{\bra}[1]{\langle #1\vert}
\newcommand{\ket}[1]{\vert#1\rangle}

\newcommand{\poly}{\textrm{poly}}

\newcommand{\eqref}[1]{Eq.~(\ref{#1})}
\newcommand{\vertiii}[1]{{\left\vert\kern-0.25ex\left\vert\kern-0.25ex\left\vert #1 
    \right\vert\kern-0.25ex\right\vert\kern-0.25ex\right\vert}}
\newtheorem{thm}{Theorem}
\newtheorem{lem}[thm]{Lemma}
\begin{document}
\date{\today}
\title{Computational complexity of time-dependent density functional theory}
\author{J. D. Whitfield$^{1,*}$, M.-H. Yung$^{2,3}$, D. G.  Tempel$^{3}$, S. Boixo$^4$, and A. Aspuru-Guzik$^3$}
 
\address{$^1$Vienna Center for Quantum Science and Technology\\University of Vienna, Department of Physics, Boltzmanngasse 5, Vienna, Austria 1190}
\address{$^2$Center for Quantum Information, Institute for Interdisciplinary Information Sciences,
Tsinghua University, Beijing, 100084, P. R. China}
\address{$^3$Department of Chemistry and Chemical Biology, Harvard University, Cambridge, Massachusetts 02138, USA}
\address{$^4$ Google, Venice Beach, CA 90292, USA}
\address{Department of Chemistry and Chemical Biology, Harvard University, Cambridge, Massachusetts 02138, USA}
\address{$^*$ Corresponding author: \ead{JDWhitfield@gmail.com}}

\begin{abstract}
Time-dependent density functional theory (TDDFT) is rapidly emerging as a premier method 
for solving dynamical many-body problems in physics and chemistry. The mathematical foundations of TDDFT 
are established through the formal existence of a fictitious non-interacting system (known as the Kohn-Sham system), which can reproduce the one-electron reduced probability density of the actual system.  We build upon these works and show that on the interior of the domain of existence, the Kohn-Sham system can be efficiently obtained given the time-dependent density.   We introduce a $V$-representability parameter which diverges at the boundary of the existence domain and serves to quantify the numerical difficulty of constructing the Kohn-Sham potential.   
For bounded values of $V$-representability, we present a polynomial time quantum algorithm 
to generate the time-dependent Kohn-Sham potential with controllable error bounds. 
\end{abstract}

\maketitle


Despite the many successes achieved so far, the major challenge of time-dependent density functional theory (TDDFT) is to find good approximations to the Kohn-Sham potential, $\hat V^{KS}$, for a non-interacting system. This is a notoriously difficult problem and leads to failures of TDDFT in situations involving charge-transfer excitations~\cite{Dreuw00}, conical intersections~\cite{Tapavicza08} or photoionization~\cite{Petersilka99}. Naturally, this raises the following question: \emph{what is the complexity of generating of the necessary potentials?}  We answer this question and show that access to a universal quantum computer is sufficient.

The present work, in addition to contributing to on-going research about the foundations of TDDFT, 
is the latest application of quantum computational complexity theory to a growing list of problems in the physics and chemistry community~\cite{Whitfield13}. Our result emphasizes that the foundations of TDDFT are not devoid of computational considerations, even theoretically. Further, our work highlights the utility of reasoning using hypothetical quantum computers to classify the computational complexity of problems. The practical implications are that, within the interior of the domain of existence, it is efficient to compute the necessary potentials using a computer with access to an oracle capable of polynomial-time quantum computation.

Quantum computers are devices which use quantum systems themselves to store and process data. On the one hand, one of the selling points of quantum computation is to have efficient algorithms for calculations in quantum chemistry and quantum physics~\cite{Brown10,Kassal11,Yung12}. On the other hand, in the worst case, quantum computers are not expected to solve all NP (non-deterministic polynomial time) problems efficiently~\cite{Bennett97}. Therefore, it is an on-going investigation into when a quantum computer would be more useful than a classical computer. Our current result points towards evidence of computational differences between quantum computers and classical computers.  In this way, we provide additional insights to one of the driving questions of information and communication processing in the past decades concerning practical application areas of quantum computing.  

Our findings are in contrast to a previous result by Schuch and Verstraete~\cite{Schuch09},
which showed that, in the worst-case, polynomial approximation to the universal functional of ground state density functional theory (DFT) is likely to be impossible even with a quantum computer.  
Remarkably, this discrepancy between the computational difficulty of TDDFT and ground state DFT is often reversed in practice where for common place systems encountered by physicists and chemists, TDDFT calculations are often more challenging than DFT calculations. Therefore, 
our findings provide more reasons why quantum computers should be built.

The practical utility of our results can be understood in multiple ways.  First, we have demonstrated a new theoretical understanding of TDDFT highlighting its relative simplicity as compared to ground state DFT computations. Second, we have introduced a $V$-representability parameter, which similar to the condition number of a matrix,  diverges as the Kohn-Sham formalism becomes less applicable. Finally, for analysis purposes, it is often useful to know what the exact Kohn-Sham potential looks like in order to compare and contrast approximations to the exchange-correlation functionals.  However, this has been limited to small dimensional or model systems and our results show that, with a quantum computer, one could perform such exploratory studies for larger systems.

\medskip
\section{Background}
\subsection{Time-dependent Kohn-Sham systems}
To introduce TDDFT and its Kohn-Sham formalism, it is instructive to view the Schr\"odinger equation as a map
\cite{Maitra10} 
\begin{equation}
	\{\hat V(t),\Psi(t_0)\}\mapsto\{n(t),\Psi(t)\}.
	\label{eq:SEmap}
\end{equation}
The inputs to the map are an initial state of $N$ electrons, $\Psi{(t=t_0)}$,  and a Hamiltonian, $\hat H(t)=\hat T+\hat W+\hat V{(t)}$ that contains a kinetic-energy term, $\hat T$, a two-body interaction term such as the Coulomb potential, $\hat W$, and a scalar time-dependent potential, $\hat V(t)$. The outputs of the map are the state at later time, $\Psi(t)$ and the one-particle probability density normalized to $N$ (referred to as the density), 
\begin{eqnarray}
\langle \hat n(x) \rangle _ {\Psi(t)}&=&\bra{\Psi(t)}\hat n(x) \ket{\Psi(t)}\nonumber\\
&=&N\int  |\Psi(x,x_2,...,x_N;t)|^2 dx_2...dx_N.
\end{eqnarray}

TDDFT is predicated on the use of the time-dependent density as the fundamental variable and all observables and properties are functionals of the density.
The crux of the theoretical foundations of TDDFT is an inverse map which has as inputs the density at all times and the initial state. It outputs the potential and the wave function at later times $t$,
\begin{equation}
	\{\langle \hat n\rangle_{\Psi(t)},\Psi(t_0)\}\mapsto\{\hat V(t),\Psi(t)\}.
	\label{eq:RGmap}
\end{equation}
This mapping exists via the Runge-Gross theorem \cite{Runge84} which shows that, apart from a gauge degree of freedom represented by spatially homogeneous variations, the potential is bijectively related to the density. However, the problem of time-dependent simulation has not been simplified; the dimension of the Hilbert space scales exponentially with the number of electrons due to the two-body interaction $\hat W$. As a result, the time-dependent Schr\"odinger equation quickly becomes intractable to solve with controlled precision on a classical computer.

Practical computational approaches to TDDFT rely on constructing the non-interacting time-dependent Kohn-Sham potential. If at time $t$ the density of a system described by potential and wave function, $\{\hat V(t),\Psi(t)\}$, is $\langle \hat n\rangle_{\Psi(t)}$, then the non-interacting Kohn-Sham system ($\hat W=0$) reproduces the same density but using a different potential, $\hat V^{KS}$. The key difficulty of TDDFT is obtaining the time-dependent Kohn-Sham potential. 

Typically, the Kohn-Sham potential is broken into three parts: $\hat V^{KS}=\hat V+\hat V^H+\hat V^{xc}$.  The first potential is the external potential given in the problem specification and the second is the Hartree potential $V^H(x,t)=\int  n(x',t)|x-x'|^{-1} d^3x'$. The third is the exchange-correlation potential and requires an approximation to be specified wherein lies the difficulty of the Kohn-Sham scheme.  In this article, we discuss how difficult approximating the full potential is but we make note that only the exchange-correlation is unknown. While we discuss the computation of the full Kohn-Sham potential from a given external potential and initial density, we will not construct an explicit functional for the exchange-correlation potential.

The route to obtaining the Kohn-Sham potentials we focus on is the evaluation of the map,
\begin{equation}
	\{\langle \hat n \rangle_{\Psi(t)}, \Phi(t_0)\}\mapsto\{\hat V^{KS}(t),\Phi(t)\}.
	\label{eq:KSmap}
\end{equation}
Here, the wave function of the Kohn-Sham system, $\Phi(t)=\mathcal A [\phi^{1}{(t)}\phi^{2}{(t)}...\phi^{N}{(t)}]$, is an anti-symmetric combination of single particle wave functions, $\phi^i(t)$, such that for all times $t$, the Kohn-Sham density, $n^{KS}{(t)}=\langle \hat n\rangle_{\Phi(t)}=\sum_{i=1}^N |\phi^i{(t)}|^2$, matches the interacting density $\langle\hat n\rangle_{\Psi(t)}$.  If such a map exists, we call the system $V$-representable while implicitly referring to non-interacting $V^{KS}$-representablity.

As the map in \eqref{eq:KSmap} is foundational for TDDFT implementations based on the Kohn-Sham system, there are many articles~\cite{Leeuwen99,Baer08,Li08,Ruggenthaler11,Ruggenthaler12,Farzanehpour12} examining the existence of such a map. Instead of attempting to merely prove the existence of the Kohn-Sham potential, we will explore the limits on the efficient computation of this map and go beyond the scope of the previous works by addressing questions from the vantage of computational complexity.

The first approach to the Kohn-Sham inverse map found in \eqref{eq:KSmap}, was due to van Leeuwen~\cite{Leeuwen99} who constructed a Taylor expansion in $t$ of the Kohn-Sham potential to prove its existence.  The construction relied on the continuity equation, $-\nabla\cdot\hat  j=\partial_t \hat n$, and the Heisenberg equation of motion for the density operator to derive the local force balance equation at a given time $t$:
\begin{equation}
	\partial_t^2\hat{n}-i[\hat W,\partial_t \hat n]=-\nabla\cdot(\hat n\nabla  V)+\hat Q,
	\label{eq:vL}
\end{equation}
where $\hat Q=i[\hat T,\partial_t \hat n]$ is the momentum-stress tensor. 
In the past few years, several results have appeared extending van Leeuwen's construction~\cite{Baer08,Li08,Ruggenthaler11,Ruggenthaler12,Farzanehpour12} 
to avoid technical problems (related to convergence and analyticity requirements). Here previous rigorous results by Farzanehpour and Tokatly~\cite{Farzanehpour12} on lattice TDDFT are directly applicable to our quantum computational setting.

\medskip
\subsection{The discrete force balance equation}
We summarize the details of the discretized local force-balance equation from~\cite{Farzanehpour12}. More detailed derivations are found in~\cite{Farzanehpour12} and as well as a more general derivation we provide in~\ref{appx:A}. 

Consider a system discretized on a lattice of $M$ points forming a 
Fock space.  In second quantization, the creation $\hat a_i$ and annihilation $\hat a_j^\dag$ operators for arbitrary sites $i$ and $j$ must satisfy $ \hat a_i\hat a_j=-\hat a_j\hat a_i$ and $\hat a_i\hat a_j^\dag=\delta_{ij}-\hat a_j^\dag \hat a_i$. We define a discretized one-body operator as $\hat A=\sum_{n}^M\sum_{m}^MA_{mn}\hat a_m^\dag \hat a_n$ and designate $A$ as the coefficient matrix of the operator. The matrix elements are $A_{mn}=\bra{m}\hat A\ket{n}$ where $\ket{m}$ and $\ket{n}$ are the single electron sites corresponding to operators $\hat a_m$ and $\hat a_n$.  Similar notation and definitions hold for the two-body operators. 

The Hamiltonian, the density at site $j$, and the continuity equation are then given respectively by 
\begin{eqnarray}
\hat H(t)&=& \sum_{ij}[T_{ij}+\delta_{ij}V_{i}(t)]\hat a_i^\dag \hat a_j+\sum_{ijkl} W_{ijkl}\hat a_i^\dag \hat a_j^\dag \hat a_k\hat a_l,\phantom{spac}\label{eq:H}\\
\hat n_j&=&  \hat a_j^\dag \hat a_j, \\
\partial_t \hat n_j&=& -\sum_k \hat J_{jk}=-i\sum_k T_{kj}(\hat a_j^\dag \hat a_k-\hat a_k^\dag \hat a_j).
\end{eqnarray}

For the density of the Kohn-Sham system, $n^{KS}(t)=\langle \hat n\rangle_{\Phi(t)}$, to match the density of the interacting system, $n(t)=\langle \hat n\rangle_{\Psi(t)}$, 
the discretized local force balance equation~\cite{Farzanehpour12} must be satisfied, 
\begin{eqnarray}
S_j^{aim}	&=&
\sum_k(V^{KS}_j-V^{KS}_k) T_{kj}\langle \hat a_j^\dag \hat a_k+\hat a_k^\dag \hat a_j\rangle_{\Phi(t)} \label{eq:1}\\
&=& \sum_k \left\langle  -T_{kj}\hat \Gamma_{jk}+\delta_{jk}\sum_m T_{mj}\hat\Gamma_{jm}\right\rangle_{\Phi(t)}V^{KS}_k\phantom{space}\label{eq:2}\\
 &=& \sum_k K_{jk}V^{KS}_k.\label{eq:K}
\end{eqnarray}
Here $\hat \Gamma_{ij}=\hat a_i^\dag \hat a_j +\hat a_j^\dag \hat a_i$ is 
twice the real part of the one-body reduced density operator.  A complete derivation of this equation is found the~\ref{appx:A}. The vector $S^{aim}$ is defined as $S^{aim}_j(\Psi,\Phi)=\partial_t^2\langle \hat{n}_j\rangle_{\Psi(t)}-\langle \hat Q^{KS}_j\rangle_{\Phi(t)}$. The force balance coefficient matrix, $K= \langle \hat K\rangle_{\Phi(t)}$, is defined through \eqref{eq:2} and \eqref{eq:K}.    
Since the target density enters only through the second derivative appearing in $S^{aim}$, the initial state $\Phi(t_0)$ must reproduce the initial density, $\langle \hat n\rangle_{\Psi(t_0)}$, and the initial time-derivative of the density, $\partial_t \langle \hat n\rangle_{\Psi(t_0)}$. 

The system is non-interacting $V$-representable so long as $K$ is invertible on the domain of spatial inhomogeneous potentials.  Moreover, the Kohn-Sham potential is unique~\cite{Farzanehpour12}. Hence, the domain of $V$-representability is $\Omega=\left\{\Phi\;|\;\textrm{kern}\; K(\Phi)=\{V_{const}\}\right\}$. To ensure efficiency, we must further restrict attention to the interior of this domain where $K$ is sufficiently well-condition with respect to matrix inversion. The cost of the algorithm grows exponentially as one approaches this boundary but can in some cases be mitigated by increasing the number of lattice points. 

\begin{figure}[t!]
\includegraphics[width =\columnwidth]{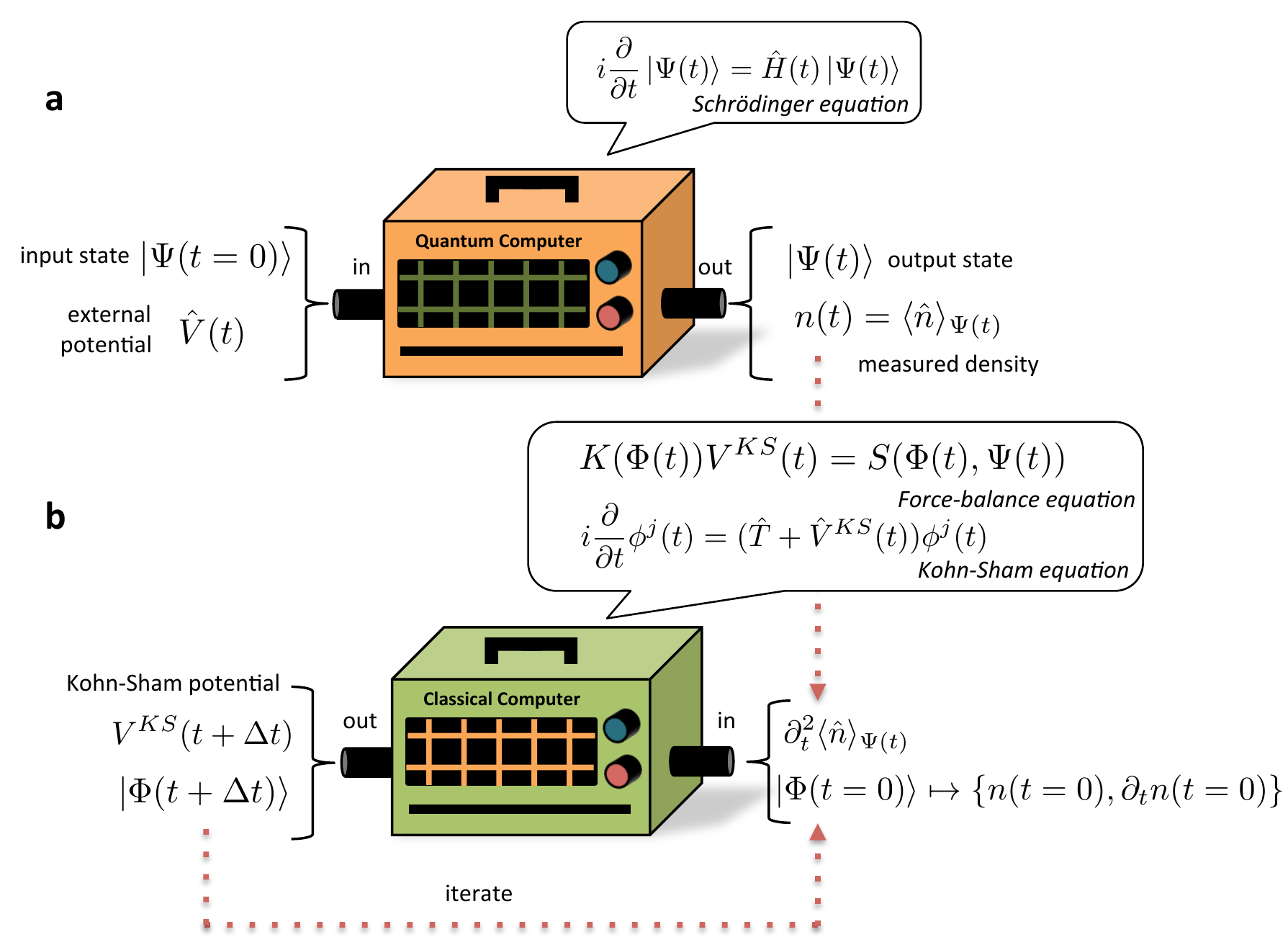}
\caption{In part \textbf{a}, the quantum computer takes as inputs the initial state and the time-dependent Hamiltonian and outputs the density at sufficiently many times. The output allows the numerical computation of the second derivative of the density at each time step which is then utilized by the classical computer to solve the discrete force balance equation \eqref{eq:K}.  A consistent initial state at time $t=0$ must also be given which reproduces $n(0)$ and $\partial_t n(0)$.  Note that while the wave function is obtained from the quantum computation, it cannot be processed for use in the classical part of the computation.   The classical algorithm uses the density to obtain the Kohn-Sham potential at each subsequent time step through an iterated marching process as depicted in part \textbf{b}.
 }
\label{fig:algorithm}
\end{figure}

\medskip
\section{Results overview}
\subsection{Quantum algorithm for the Kohn-Sham potentials}
We consider an algorithm to compute the density with error $\epsilon$ in the 1-norm to be efficient 
when the temporal computational cost grows no more than polynomially in $1/\epsilon$, polynomially in $(\max_{0<s<t}\|H(s)\|)t$, polynomially in $M$, the number of sites, and polynomially in, $N$, the number of electrons. We will describe such an algorithm within the interior of the domain of $V$-representability.

To ensure that the algorithm is efficient, we must assume that the local kinetic energy and the local potential energy are both bounded by constant $E_L$ and that there is a fixed number, $\kappa$ such that  $ \|K^{-1}\|_{\infty}=\max_i \sum_j |(K^{-1})_{ij}|\leq \kappa$.  Note that, as we work in the Fock space, this condition does not preclude Coulombic interactions with nuclei so long as the site orbitals have finite spatial extent. 

We will show that as long as $E_L\leq\sqrt{\log N}$, the algorithm remains efficient for fixed $\kappa$. 
As is typical in numerical matrix analysis~\cite{Horn05,Golub13},  the inversion of a matrix become extremely sensitive to errors as the condition number, $C=\|K\|\;\|K^{-1}\|$, grows. The Lipschitz constant of the Kohn-Sham potential must also scale polynomially with the number of electrons. 

The Lipschitz constant of the Kohn-Sham system could be different than that of the interacting system~\cite{Elliott12,Maitra10} and understanding of the relationship between these timescales requires a better understanding of the initial state $\Phi(t_0)$ dependence. What can be done, in practice, is to begin with an estimate of the maximum Lipschitz constant and if any two consecutive Kohn-Sham potentials violate this bound, restart with a larger Lipschitz constant. 

Our efficient algorithm for computing the time-dependent potential, is depicted in Figure~\ref{fig:algorithm}. There are two stages. The first stage involves a quantum computer and its inputs are the initial many-body state $\Psi(t_0)$ and the external potential $V(t)$ on a given interval $[t_0,t_1]$. The quantum computer then evolves the initial state with the given external potential and obtains the time-evolved wave function at a series of discrete time-steps. The detailed analysis of the expectation estimation algorithm found in Ref.~\cite{Knill07} is used to bound errors in the measurement of the density and to estimate its second time derivative. In order to rigorously bound the error term, we assume that the fourth time derivative of the density  is bounded by a constant, $c_4$.

The total cost of both stages of the algorithm is dominated by the cost of obtaining the wave function as this is the only step that depends directly on the number of electrons. Fortunately, quantum computers can perform time-dependent simulation efficiently~\cite{Wiebe10,Poulin11,Berry13}. The cost depends on the requested error in the wave function, $\delta_\psi$, and depends on the length of time propagated when time is measured relative to the norm of the Hamiltonian being simulated. The essential idea is to leverage the evolution of a controllable system (the quantum computer) with an imposed (simulation) Hamiltonian~\cite{Kassal11}.   It should be highlighted that obtaining the density through experimental spectroscopic means is equivalent to the quantum computation provided the necessary criteria for efficiency and accuracy are satisfied.

The second stage involves only a classical computer, with the inputs being a consistent initial Kohn-Sham state $\Phi(t_0)$ and the interacting $\partial_t^2 \langle \hat n\rangle_{\Psi(t)}$ on the given interval $[t_0,t_1]$.  The output is the Kohn-Sham potential at sufficiently many time steps to ensure the target accuracy is achieved.
The classical algorithm performs matrix inversion of a $M$ by $M$ matrix. The cost for the matrix inversion is $O(M^3)$ regardless of the other problem parameters (such as the number of electrons). 

In our analysis detailed in the next section, we only consider errors from the quantum and classical aspects of our algorithm and we avoided some unnecessary complications by omitting detailed analysis of the classical problem of propagating the non-interacting Kohn-Sham system. Kohn-Sham propagation in the classical computer is well studied and can be done efficiently using various methods \cite{Castro04}. Further, we have also assumed that errors in the measured data are large enough that issues of machine precision do not enter.  Thus, we have ignored the device dependent issue of machine precision in our analysis and refer to standard treatments~\cite{Horn05,Golub13} for the proper handling of this issue.

\medskip
\subsection{Overview of error bounds}
We demonstrate that our algorithm has the desired scaling by bounding the final error in the density. We follow an explicit-type marching process to obtain the solution at time $q\Delta t$ from the solution at $(q-1)\Delta t$. The full technique is elaborated in the next section.

As the classical matrix inversion algorithm at each time step is independent of the number of electrons and the quantum algorithm requires $\poly(N,t_1-t_0,\delta_\psi^{-1},\epsilon^{-1})$ per time step (recall that $\delta_\psi$ is the allowed error in the wave function due to the quantum simulation algorithm), we can utilize error analysis for matrix inversion and an explicit marching process to get a final estimate of the classical and quantum costs for the desired precision $\epsilon$
\begin{eqnarray}
	\textrm{cost }&Classical&=\textrm{poly}(L,t_1-t_0,\epsilon^{-1},M){e^{64 \kappa E_L^2}}\label{eq:13}\\
\textrm{cost }&Quantum&=\textrm{poly}(L,t_1-t_0,\epsilon^{-1},r,M,N)\; {e^{16\kappa E_L^2}}\label{eq:14}\phantom{spc}
\end{eqnarray}
The parameter $r$ is the number of repetitions of the quantum measurement required to obtain a suitably large confidence interval.  We define the $V$-representability parameter as $R=\kappa E_L^2$ and if $R$ is bounded by a constant, then the algorithm is efficient.

The intractability of the algorithm with growing $R$ indicates the breakdown of $V$-representability. 
Despite the exponential dependence of the algorithm on the representability parameter,
the domain of $V$-representability is known to encompass all time-analytic Kohn-Sham potentials in the continuum limit~\cite{Baer08,Li08,Ruggenthaler11,Ruggenthaler12}.  Examining the exponential dependence, it is clear that increases in $\kappa$ can be offset by decreases in the local energy. 

\section{Derivation of error bounds}
\subsection{Description of techniques used to bound cost}
Before diving into the details, let us give an overview of our techniques and what is to follow. In the first subsection, we look at the error in the wave function at time $t$.  In each time step, the error is bounded from the errors in the previous steps. This leads to a recursion relation which we solve to get a bound for the total error at any time step. This error is propagated forward because we must solve $KV=S=Q+\partial_t^2 n$ for $V$ based on the data from the previous time step.  The error in $\partial_t^2 n$ is due to the finite precision of the quantum computation and is independent of previous times. In the second subsection, the error in the density is then derived followed by a cost analysis in the final subsection.

We rescale time by factor $c$ such $t_1-t_0=1$ to get the final time step $z=1/\Delta t$. This rescaling is possible because there is no preferred units of time.
That said the rescaling of time cannot be done indefinitely for two reasons. First, the Lipschitz constant of both the real and the KS system must be rescaled by same factor of $c$. Since the cost
of the algorithm depends on the Lipschitz constant, increasingly long times will require more resources.  Second, the quantum simulation algorithm does have an intrinsic time scale
set by the norm of the $H$ and its time derivatives~\cite{Wiebe10,Poulin11,Berry13}. Rescaling time by $c$ increases the norm of $H$ by the same factor; consequently, the difficulty of 
the quantum simulation is invariant to trivial rescaling of the dynamics.

It is important to get estimates which do not directly depend on the number of sites. To do this, we assume that the lattice is locally connected under the hopping term such that there are at most $d$ elements per row of $T$ (since $T$ is symmetric, it is also $d$-col-sparse). This is equivalent to a bound for the local kinetic energy.

Throughout, we work with the matrix representations of the operators and the states.  
The $L_p$ vector norms \cite{Horn05} with $p=1$, $2$, and $\infty$ are defined by $|x|_p=\left(\sum |x_i|^p  \right)^{1/p}$.   
The induced matrix norms are defined by $\|A\|_p=\max_{|x|_p=1}|Ax|_p$. Induced norms are important because they are compatible with the vector norm such that $|Mx|_p=\|M\|_p|x|_p$.
The vector 1-norm is appropriate for probability distributions and the vector 2-norm is appropriate for wave functions.  The matrix 2-norm is also called the spectral norm and is equal to the maximum absolute value of an eigenvalue.  For a diagonal matrix, $D$, the matrix 2-norm is the vector $\infty$-norm of $\textrm{diag}(D)$. Note that $|x|_p\geq|x|_{p'}$ for $p<p'$. Important, non-trivial characterizations of the infinity norms are $|x|_\infty=\max_i |x_i|$ and  $\|A\|_\infty=\max_i \sum_j |A_{ij}|$.

\subsection{Error in the wave function via recursion relations}
We bound the error of the evolution operator from time $k\Delta t$ to $(k-1)\Delta t$, denoted $\|\Delta U({k,k-1)}\|_2$, in terms of the previous time step in order 
to obtain a recursion relation.  We first bound the errors in the potential 
due to the time discretization and then those due to the computation errors using Lemma \ref{lem:unitary} found in~\ref{C}.   The computation errors will depend on the error at the previous time step which will lead to the recursion relation sought after.

To bound the error in $\|\Delta U\|_2$ we must bound the error in the potential 
$|\Delta V|_\infty\leq|\Delta V^{\Delta t}|_\infty+|\Delta V^{comp}|_\infty$. We define $ V^{\Delta t}(t)=V({t_k})$ with 
$k$ such that $|t-t_k|\leq |t-t_m|$ for all $m$. Here, $\{V(t_k)\}$ is the discretized potential 
with time step $|t_j-t_{j+1}|=\Delta t$. The error due to temporal discretization can be controlled 
assuming a Lipschitz constant $L$ for the potential such that for all $t$ and $t'$,
$|V(t)-V({t'})|_\infty/|t-t'|\leq L$.  Thus, for all $t$, 
\begin{equation}
|\Delta V^{\Delta t}|_\infty=|V(t)-V^{\Delta t}(t)|_\infty\leq L\Delta t.
\end{equation} 

The computational error 
$|\Delta V^{comp}|_\infty$ is bounded using Lemma \ref{lem:linalg} in~\ref{C} with $\|K^{-1}\|_\infty\leq \kappa$ and the assumption $|V|_\infty\leq E_L$,
\begin{equation}
	|\Delta V^{comp}|_\infty\leq \kappa \left( |\Delta Q|_\infty +|\Delta \partial_t^2 n|_\infty+\|\Delta K\|_\infty E_L \right)
	\label{eq:comp}
\end{equation}
Now we need to bound the errors in $|\Delta Q|_\infty$ and $\|\Delta K\|_\infty$ in terms of the error $\delta^\Gamma_k=\max_{ij}|\Delta \Gamma_{ij}(k-1)|$ at time step $k-1$.  

The error bound for $|\Delta Q|_\infty$ is obtained as  
\begin{eqnarray}
	|\Delta Q|_\infty&\leq&	\max_i|([T,\Delta\Gamma]T)_{i}|\\
	&\leq&\max_i \left|\sum_{pq}T_{ip}\Delta\Gamma_{pq}T_{qi}-\sum_{mn}\Delta\Gamma_{im}T_{mn}T_{ni}\right|\nonumber\\
&\leq&
2\delta_{k-1}^\Gamma d^2 \left(\max_{ij}|T_{ij}|\right)^2\nonumber\\
|\Delta Q|_\infty&\leq&2\delta_{k-1}^\Gamma E_L^2
\label{eq:Q}
\end{eqnarray}
The product $d\max |T_{ij}|$ is the maximum local kinetic energy and is, by assumption, bounded by $E_L$. Similarly,
\begin{eqnarray}
\|\Delta K\|_{\infty}&=&\max_i\sum_j|K_{ij}-\tilde{K}_{ij}|\\
&=&
\max_i\sum_j  |T_{ij}\Delta\Gamma_{ij}-\delta_{ij}\sum_m T_{mj}\Delta\Gamma_{mj}|\nonumber\\
&\leq& \max_i\sum_j  |T_{ij}\Delta\Gamma_{ij}|+\max_i\left|\sum_m T_{mi}\Delta\Gamma_{mi}\right|\nonumber\\
&\leq&\delta_{k-1}^\Gamma \max_i\sum_j  |T_{ij}|+\delta_k^\Gamma\max_i\left|\sum_m T_{mi}\right|\nonumber\\
&\leq&2d\delta_{k-1}^\Gamma \left( \max_{ij}|T_{ij}| \right)\nonumber\\
\|\Delta K\|_{\infty}&\leq&2\delta_{k-1}^\Gamma E_L\label{eq:Kerr}
\end{eqnarray}
We convert from errors in the real part of the 1-RDM to errors in the wave function via
\begin{eqnarray}
	\delta^{\Gamma_{ij}}&=& |\Delta\Gamma_{ij}|\nonumber\\
		&\leq& |(\bra{\Phi}\Gamma_{ij})\ket{\Delta\Phi}|+|\bra{\Delta\Phi}(\Gamma_{ij}\ket{\Phi})|\label{eq:expn}\\
		&\leq& 2|\Delta\Phi|_2\;|\Gamma_{ij}\ket{\Phi}|_2\leq 2 |\Delta\Phi|_2\;\|\Gamma_{ij}\|_2\nonumber\\
		&\leq&4|\Delta\Phi|_2\label{eq:ev}
\end{eqnarray}
The inequality  \eqref{eq:ev} follows because the maximum eigenvalue of $\langle a_i^\dag a_j\rangle_\psi$ for all $\psi$ is bounded by $1$ and $\Gamma_{ij}=2\textrm{ real}\langle a_i^\dag a_j\rangle_\psi$. 
Taking the maximum over all $i$, $j$ we have 
\begin{equation}
	\delta_{k-1}^\Gamma=\max_{ij}	(\delta^{\Gamma_{ij}}_{k-1})	\leq4\delta_{k-1}^\Phi\label{eq:G}
	\end{equation}
 Here $\delta_{k-1}^\Phi$ bounds the error in the two-norm $|\Delta \Phi|_2$ at time step $k-1$.

Putting together \eqref{eq:comp}, \eqref{eq:Q}, \eqref{eq:Kerr}, and \eqref{eq:G} gives
\begin{eqnarray}
	|\Delta V^{comp}|_\infty\leq16\kappa E_L^2\delta_{k-1}^{\Phi}+\kappa|\Delta \partial_t^2n|_\infty
	\label{eq:comp2}
\end{eqnarray}

To obtain the desired recursion relation, we note that at time step $k$ the error can be bounded via
\begin{eqnarray}
|\Phi(k)-\tilde{\Phi}(k)|_2&\leq& \|\Delta U({k,k-1})\|_2
+\delta_{k-1}^\Phi
\end{eqnarray}
obtained using an expansion similar to the one found in \eqref{eq:expn}. 
Utilizing Lemma \ref{lem:unitary} (see~\ref{C}) and bound \eqref{eq:comp2}, we arrive at 
\begin{eqnarray}
	|\Phi(k)-\tilde{\Phi}(k)|_2
	&\leq&\delta^\Phi_{k-1}+\Delta t|\Delta_{k,k-1}V|_\infty\nonumber\\
	&\leq&\delta^\Phi_{k-1}+\Delta t(|\Delta V^{\Delta t}|_\infty+|\Delta V^{comp}|_\infty)\nonumber\\
	&\leq&\delta^\Phi_{k-1}+\Delta t(L\Delta t+16\kappa E_L^2\delta_{k-1}^{\Phi}+\kappa|\Delta \partial_t^2n|_\infty)\nonumber\\
	&\leq&(16\kappa E_L^2\Delta t+1)\delta_{k-1}^{\Phi}\nonumber\\
	&&+\Delta t(L\Delta t+\kappa|\Delta \partial_t^2n|_\infty)\label{eq:prerecus}
\end{eqnarray}
To obtain a recursion relation we let the LHS of \eqref{eq:prerecus} define the new upper bound at time step $k$. 

Recursion relations of the form $f_k=af_{k-1}+b$ have closed solution
$f_k=b(a^k-1)(a-1)^{-1}$. Thus, we have for the bound at time step $k$
\begin{eqnarray}
	\delta_k^\Phi	
	&=&\frac{L\Delta t+\kappa|\Delta \partial_t^2n|_\infty}{16\kappa E_L^2} \left\{(16\kappa E_L^2\Delta t+1)^k-1\right\}
	\end{eqnarray}
Now consider the final time step at $z=1/\Delta t$, and $e^x\geq(xz^{-1}+1)^z$ for $z<\infty$,
	\begin{eqnarray}
		\delta_z^\Phi&=&\frac{L\Delta t+\kappa|\Delta \partial_t^2n|_\infty}{16\kappa E_L^2} \left\{\left(\frac{16\kappa E_L^2}{z}+1\right)^z-1\right\}\\
		&\leq&\left( \frac{1}{z}\frac{L}{16\kappa E_L^2}+\frac{|\Delta \partial_t^2n|_\infty}{16E_L^2}\right) \left\{e^{16\kappa E^2_L}-1\right\}\\
		&\leq&\left( \frac{1}{z}\frac{L}{16\kappa E_L^2}+\frac{\sqrt{2c_4\delta_n}}{16E_L^2}\right) \left\{e^{16\kappa E^2_L}-1\right\}\label{eq:wferr}
	\end{eqnarray}
We applied Lemma \ref{lem:ddot} from~\ref{C} to obtain the last line. 
This bound is similar to the Euler formula for the global error but arises from the iterative 
dependence of the potential on the previous error; not from any approximate solution to
an ordinary differential equation. 

To ensure that the cost is polynomial in $M$ and $N$ for fixed $\kappa$, we must insist that $E_L\leq \sqrt{\log N}$.  Consider the exponential factor 
and assume that $E_L>1$. Then $\exp(16\kappa E^2_L)\leq \exp(16\kappa \log N)=N^{16\kappa}$ is a polynomial for fixed $\kappa$.

\subsection{Error bound on the density}
To finish the derivation, we utilize our bound for the wave function at the final time to get a bound on the error of the density at the final time.  This will translate into conditions for the number of steps needed and the precision required for the density. 
The error in the density is bounded by the error in the wave function through  the following,
\begin{eqnarray*}
|\Delta n|_1&=& |\bra{\Phi}n\ket{\Phi}-\bra{\tilde \Phi}n\ket{\tilde\Phi}|_1\\
&=&|\bra{\Phi}n\ket{\Phi}-\bra{\Phi}n\ket{\tilde \Phi} +\bra{\Phi}n\ket{\tilde \Phi} -\bra{\tilde \Phi}n\ket{\tilde\Phi}|_1\\
&\leq&|\bra{\Phi}n\ket{\Delta\Phi} |_1+|\bra{\Delta\Phi}n\ket{\Phi}|_1
\end{eqnarray*}
Now consider the $i$-th element, $n_i=a_i^\dag a_i$, and the Cauchy-Schwarz $|\bra{x}y\rangle| \leq|x|_2\;|y|_2$,
\begin{eqnarray*}
	\left|\left(\bra{\Phi}a_i^\dag a_i\right)\ket{\Delta\Phi}\right|&\leq& \left|\bra{\Phi}a_i^\dag a_i\right|_2\; \left|\Delta\Phi\right|_2
	\leq\|a_i^\dag a_i\|_2\; |\Delta\Phi|_2\\
|\bra{\Phi}n_i\ket{\Delta\Phi} |_1	&\leq &|\Delta\Phi|_2
\end{eqnarray*}
Finally, from the definition of the 1-norm,
\begin{eqnarray}
	|\Delta n(z)|_1&\leq& \sum_i \left(|\bra{\Delta\Phi(z)}n_i\ket{\tilde\Phi(z)}|+|\bra{\Phi(z)}n_i\ket{\Delta\Phi(z)}|\right)\nonumber\\&\leq&2M|\Delta\Phi(z)|_2\leq2M\delta_z^\Phi
	\label{eq:err1}
\end{eqnarray}

For final error $\epsilon$ in the 1-norm of the density, we allow error $\epsilon/2$ due to the time step error and 
$\epsilon/2$ error due to the density measurement. Following \eqref{eq:wferr} and \eqref{eq:err1}, we have for the number of time steps,
\begin{equation}
	\left( \frac{ML}{4\epsilon\kappa E_L^2}\right) \left\{e^{16\kappa E^2_L}-1\right\}	\leq z.
	\label{eq:z}
\end{equation}
The bound for the measurement precision also follows as,
\begin{equation}
	\left( \frac{\sqrt{2}Mc_4^{1/2}}{4\epsilon E_L^2}\right)^{2} \left\{e^{16\kappa E^2_L}-1\right\}^{2}	\leq\delta_n^{-1}
	\label{eq:n}
\end{equation}

\subsection{Cost analysis}
To obtain the cost for the quantum simulation and the subsequent measurement, we leverage detailed analysis of the expectation estimation algorithm \cite{Knill07}. To measure the density at time $t\in[t_0, t_1]$, a quantum simulation~\cite{Wiebe10,Poulin11,Berry13}  of $\psi({t_0})\mapsto\psi(t)$ is performed at cost $q\leq \poly(N,t_1-t_0,\delta_\psi^{-1})$ following an assumption that $H(t)$ is simulatable on a quantum computer which is usually the case for physical systems.  In order to simplify the analysis, we assume that $\delta_\psi$ is such that $\delta_n+\delta_\psi\approx\delta_n$ is a reasonable approximation. Given the recent algorithm for logarithmically small errors \cite{Berry13}, this assumption is reasonable.   

The expectation estimation algorithm (EEA) was analyzed in \cite{Knill07}. The algorithm EEA$(\psi,A,\delta,c)$ measures $\bra\psi A \ket \psi$ with precision $\delta$ and confidence $c$ such that Prob$(\tilde a-\delta\leq \bra\psi A \ket \psi\leq \tilde a+\delta)>c$ , that is, the probability that the measured value $\tilde a$ is within $\delta$ of $\bra\psi A\ket\psi$ is bounded from below by $c$. The idea is to use an approximate Taylor expansion:  $$\bra\psi A \ket \psi\approx i\left( \bra\psi e^{-iAs}\ket\psi -1 \right)/s$$
The confidence interval is improved by repeating the protocol $r=|\log(1-c)|$ times.  If the spectrum of $A$ is bounded by $1$, then the algorithm requires on the order $O(r/\delta^{3/2})$ copies of $\psi$ and $O(r/\delta^{3/2})$ uses of $\exp(-iAs)$ with $s=\sqrt{3\delta}/2$.

To perform the measurement of the density, we assume that the wave function is represented in first quantization~\cite{Kassal11} such that the necessary evolution operator is: $\exp(-i\hat n_j s)=\prod_k^N\exp(-i \ket{j}\bra{j}^{(k)}t)$. Here each Hamiltonian $\ket{j}\bra{j}^{(k)}$ acts on site $j$ of the $k$th electron simulation grid. Hence, each operation is local with disjoint support.  Since there are $NM$ sites, this can be done efficiently.  Comparing the costs, we will assume that the generation of the state dominates the cost.

Combining these facts, we arrive at the conclusion that the cost to measure the density to within $\delta_n$ precision is 
\begin{eqnarray}
	\textrm{cost }{Quantum}&=& \textrm{cost }StateGen+\textrm{cost }EEA\nonumber\\
	&\approx& \textrm{cost }StateGen\nonumber\\
&=& O\left(rq\delta_n^{-3/2}\right)\phantom{spc}
\label{eq:q0}
\end{eqnarray}
Pairing this with \eqref{eq:z} and \eqref{eq:n}, we have an estimate for the number of quantum operations 
\begin{eqnarray*}
\textrm{cost }	Quantum &=&  O\left(rqz\delta_n^{-3/2}\right)\\
&=&\textrm{poly}(L,\epsilon^{-1},r,M,N)\; {e^{64\kappa E_L^2}}
\end{eqnarray*}
The classical computational algorithm is an $[M\times M]$ matrix inversion at each time step costing
\begin{eqnarray*}
 \textrm{cost } Classical &=& O(z M^3)\\
 &=&O\left(M^3\left( \frac{ML}{4\epsilon\kappa E_L^2}\right) \left\{e^{16\kappa E^2_L}-1\right\}\right)\\
&=& \textrm{poly}(L,\epsilon^{-1},M)e^{16\kappa E^2_L}
\end{eqnarray*}

\medskip
\section{Quantum computation and the computational complexity of TDDFT}
Since the cost of both the quantum and classical algorithms scale as a polynomial of the input parameters, we can say that this is an efficient quantum algorithm for computing the time-dependent Kohn-Sham potential. Therefore, the computation of the Kohn-Sham potential is in the complexity class described by bounded error quantum computers running in polynomial time (BQP). This is the class of problems that can be solved efficiently on a quantum computer.  

Quantum computers have long been considered as a tool for simulating quantum physics~\cite{Feynman82,Lloyd96,Brown10,Kassal11,Yung12}. The applications of quantum simulation fall into two broad categories: (1) dynamics~\cite{Zalka98,Lidar99,Kassal08} and (2) ground state properties~\cite{Somma02,Aspuru05,Whitfield11}.  The first problem is in the spirit of the original proposal by Feynman~\cite{Feynman82} and is the focus of the current work.  

Unfortunately, unlike classical simulations, the final wave function of a quantum simulation cannot be readily extracted due to the exponentially large size of the simulated Hilbert space. The retrieval of the full state would require quantum state tomography, which in the worst case, requires an exponential number of copies of the state and would take an exponentially large amount of space to even store the data classically.  If, instead, the simulation results can be encoded into a minimal set of information and the simulation algorithm can be efficiently executed on a quantum computer, then the problem is in the complexity class BQP. Extraction of the density~\cite{Knill07} is the relevant example of such a quantity that can be obtained.  Note that the density's time-evolution is dictated by wave function and hence the Schr\"odinger equation. 


In summary, what we have proven is that computing the Kohn-Sham potential at bounded $\kappa E_L^2$ is in the complexity class BQP.  To be precise, two technical comments are in order. First, we point out that we are really focused on promise problems since we require constraints on the inputs to be satisfied (i.e. $\kappa E_L^2<$constant). Second, computing the map \eqref{eq:KSmap} is not a decision problem and cannot technically be in the complexity class BQP.  However, we can define the map to $b$ bits of precision by solving $M\log b$ accept-reject instances from the corresponding decision problem, which is in BQP. These concepts are further elaborated in \cite{Kitaev02,Watrous09,Whitfield13}. 

While the quantum computer would allow most dynamical quantities to be extracted without resorting to the Kohn-Sham formalism, we have attempted to understand the difficulty of generating the Kohn-Sham potential.  We only consider a polynomial time quantum computer as a tool for reasoning about the complexity of computing Kohn-Sham potentials.  In essence, the Kohn-Sham potentials are a compressed classically tractable encoding of the quantum dynamics that allows the quantum simulation to be performed in polynomial time on a classical computer.  This may have implications for the question of whether a classical witness can be used in place of quantum witness in the quantum Merlin Arthur game~\cite{Watrous09} (i.e. QMA$\stackrel{?}{=}$QCMA).  A second useful by-product of our result is the introduction of the $V$-representability parameter which has general significance for practical computational settings.


\medskip
\section{Concluding remarks}
In this article, we introduced a $V$-representability parameter and have rigorously demonstrated two fundamental results concerning the computational complexity of time dependent density functional theory with bounded representability parameter.  First, we showed that with a quantum computer, one need only provide the initial state and external potential on the interval $[t_0, t_1]$ in order to generate the time-dependent Kohn-Sham potentials. Second, we show that if one provides the density on the interval $[t_0, t_1]$, the Kohn-Sham potential can be obtained efficiently with a classical computer.

We point out that an alternative to our lattice approach may exist using tools from partial differential equations.  Early results in this direction have been pioneered using an iterated map whose domain of convergence defines $V$-representability~\cite{Ruggenthaler11,Ruggenthaler12}. The convergence properties of the map have been studied in several one-dimensional numerical examples~\cite{Nielsen13, Ruggenthaler11,Ruggenthaler12}. Analytical understanding of the rate of convergence to the fixed point would complement the present work with an alternate formulation directly in real space.

While this paper focuses on the simulation of quantum dynamics, the complexity of the ground state problem is interesting in its own right \cite{Kitaev02,Watrous09,Whitfield13, Schuch09}. In this context, ground state DFT was formally shown \cite{Schuch09} to be difficult even with polynomial time quantum computation.  Interestingly,
in that work, the Levy-minimization procedure \cite{Levy79} was utilized for the interacting system to avoid discussing the non-interacting ground state Kohn-Sham system and its existence. We have worked within the Kohn-Sham picture, but it may be interesting to construct a functional approach directly.

Future research involves improving the scaling with the condition number or showing that our observed exponential dependence on the representability parameter is optimal. Our work can likely be extended to bosonic and spin systems~\cite{Tempel12} since we have relied minimally on the fermionic properties of electrons.  Finally, pre-conditioning the matrix $K$ can also help increase the domain of computationally feasible $V$-representability.

Our findings provide further illustration of how the fields
of quantum computing and quantum information can contribute to our understanding of
physical systems through the examination of quantum complexity theory.

\subsection*{Acknowledgements:} We appreciate helpful discussions with F. Verstraete and D. Nagaj.  JDW thanks Vienna Center for Quantum Science and Technology for the VCQ Postdoctoral Fellowship and acknowledges support from the Ford Foundation. MHY acknowledges funding support from the National Basic Research Program of China Grant 2011CBA00300, 2011CBA00301, the National Natural Science Foundation of China Grant 61033001, 61061130540. MHY, DGT, and AAG acknowledge the National Science Foundation under grant CHE-1152291 as well as the Air Force Office of Scientific Research under grant FA9550-12-1-0046. AAG acknowledges generous support from the Corning Foundation. 

\section*{Bibliography}
\bibliography{tddft.bib}
\bibliographystyle{unsrt}
\newpage
\appendix

\section{Derivation of discrete local-force balance equation}\label{appx:A}
The results found in Farzanehpour and Tokatly~\cite{Farzanehpour12}, are directly applicable to the quantum computational case since a quantum simulation would ultimately require a discretized space~\cite{Kassal11}.  In~\cite{Farzanehpour12}, they utilized a discrete space but derive all equations in first quantization. For this reason, we think the derivation in second quantization may be useful for future inquiries into discretized Kohn-Sham systems and provide the necessary details in this appendix. Throughout this section, we consider the non-interacting Kohn-Sham system without an interaction term, i.e.~$\hat W=0$.

First note,
$[\hat a_p^\dag \hat a_q,\hat a_j^\dag]=\hat a_p^\dag\delta_{jq}$ and 
$[\hat a_p^\dag \hat a_q,\hat a_i]     =-\hat a_q\delta_{ip}$  to get the first derivative of the density 
\begin{eqnarray}
\partial_t\hat {n}_j&=&	-\sum_k \hat J_{jk}=i[\hat H,\hat n_j]\\
	&=&i\sum_{pq} T_{pq}[\hat a_p^\dag \hat a_q,\hat a_j^\dag \hat a_j] \\
	&=& -i\sum_{k}T_{kj}  (\hat a_j^\dag \hat a_k-\hat a_k^\dag \hat a_j)	
\end{eqnarray}
Here and throughout, we assume that there is no magnetic field present and consequently $T_{ij}=T_{ji}$. 

To get to the discrete force balance equation, consider $\partial_t^2\hat{n}_j=i[\hat H,\partial_t\hat{n}_j]=i[\hat V,\partial_t\hat{n}_j]+\hat Q_j+i[\hat W,\partial_t\hat{n}_j]$ with 
$\hat Q_j=i[\hat T,\partial_t\hat{n}_j ]$, a term that does not depend on the local potential.
This is analogous to \eqref{eq:vL} first derived in van Leeuwen's paper~\cite{Leeuwen99}.  

In the case that the non-interacting Kohn-Sham potential is desired, only the momentum-stress tensor is needed since $\hat W=0$ in the non-interacting system.  We will need the expression for $\hat Q_j$ so let us compute it now for the KS system,
\begin{eqnarray}
\hat Q_j&=& i[\hat T,\partial_t \hat n_j]= \sum_{pq} \sum_k T_{pq}T_{jk}[\hat a_p^\dag \hat a_q,\hat a_j^\dag \hat a_k-\hat a_k^\dag \hat a_j]\\
&=&\sum_{pq} \sum_k T_{pq}T_{jk}
	(\hat a_p^\dag \hat a_k +\hat a_k^\dag \hat a_p) \delta_{jq}
-\sum_{pq} \sum_k T_{pq}T_{jk}	(\hat a_j^\dag \hat a_p+\hat a_p^\dag \hat a_j) \delta_{qk}
\\
&=&\sum_{pq} \sum_k T_{pq}T_{jk}
\left\{
	 \hat \Gamma_{kp} \delta_{jq}- \hat \Gamma_{jp} \delta_{qk}
\right\}\\
&=& \sum_{pq}  T_{pq}\delta_{jq}\left(\sum_kT_{jk} \hat \Gamma_{kp}\right) - \sum_{qk}  T_{jk}\delta_{qk}\left(\sum_p \hat \Gamma_{jp}T_{pq}\right)\\
&=& \sum_p\left(\sum_k T_{jk} \hat \Gamma_{kp}\right)T_{pj}- \sum_{q} \left(\sum_p \hat \Gamma_{jp}T_{pq}\right)T_{qj}\\
&=&\left(\phantom{|^A_B}  \left[T,\hat \Gamma\right]\;T\phantom{|^A_B}  \right)_{jj}
\end{eqnarray}
Here we have defined the real part of the 1-RDM as $\hat \Gamma_{ij}= \hat a_i^\dag \hat a_j+\hat a_j^\dag \hat a_i$ following the notation in the main text and $T$ is the coefficient matrix of the kinetic energy operator.

Next, we obtain more convenient 
representations for the local force balance equation.  Beginning with $\partial_t^2 \hat n=i[\hat H,\partial_t \hat n]=i[\hat T,\partial_t \hat n]+i[\hat V,\partial_t \hat n]
=\hat Q+i[\hat V,\partial_t \hat n]$.
 Defining $\hat S=\partial_t^2 \hat n-\hat Q$, 
 we have the following,
\begin{eqnarray}
	 \hat S_j	&=& i[\hat V,\partial_t\hat {n}_j]\nonumber
	=i\left[\left(\sum_m V_m \hat a_m^\dag \hat a_m\right),\left(-i\sum_{k}T_{kj}  (\hat a_j^\dag \hat a_k-\hat a_k^\dag \hat a_j)	\right)\right]\nonumber\\
	&=&   
	       \sum_{k} V_j  T_{kj}\hat a_j^\dag\hat a_k
	       +\sum_{k} V_j  T_{kj}\hat a_k^\dag\hat a_j
	-\sum_{k} V_k  T_{kj}\hat a_j^\dag\hat a_k
	       -\sum_{k} V_k  T_{kj}\hat a_k^\dag\hat a_j\nonumber\\
       &=& 
       \sum_k(V_j-V_k) T_{kj}(\hat a_j^\dag\hat a_k+\hat a_k^\dag\hat  a_j )\label{eq:1a}\\
	&=& 
	  \sum_m T_{mj}( \hat a_j^\dag \hat a_m+\hat a_m^\dag \hat a_j) \left(\sum_k\delta_{jk}V_k\right)\nonumber
       	 -\sum_k T_{kj}( \hat a_j^\dag \hat a_k+\hat a_k^\dag \hat a_j) V_k\nonumber\\
	 &=&
       	 \sum_k\left\{-T_{kj}\hat\Gamma_{jk}+\delta_{jk}\sum_m T_{mj}
	\hat \Gamma_{jm}	 
	 \right\}V_k\label{eq:2a}
\end{eqnarray}
So now consider the LHS as vector $\hat S$ with components $\hat S_j=\partial_t^2\hat{n}_j-\hat Q_j$. Similarly consider the potential 
$V$ as a vector with components $V_i$, then we can write  \eqref{eq:2a} as $\hat S=\hat K V$. 
Examining \eqref{eq:1a}, if $V_k=V_{k'}$ for all $k,k'$ then the RHS of \eqref{eq:1a} vanishes.  Hence, $K$ always has at least one vector in the null space, namely the spatially constant potential.

Farzanehpour and Tokatly \cite{Farzanehpour12} study the existence of a unique solution for the non-linear Schr\"odinger equation which follows from \eqref{eq:2a}: 
\begin{equation}
	\partial_t\Phi=-i( \hat H_0+\hat V^{KS})=-i(\hat  H_0- \hat K(\Phi)^{-1}\hat S)\Phi=\hat F(\Phi)\;.
	\label{eq:NLSE}
\end{equation} 
In the space where $\hat K$ has only one zero eigenvalue, the Picard-Lindel\"of theorem~[M. E. Lindel\"of, C. R. Hebd. Sances Acad. Sci. \textbf{116}, 454 (1894)] guarantees the existence of a unique solution.

The Picard-Lindel\"of theorem concerns the differential equation $\partial_t\, y(t)=f(t,y(t))$ with initial value $y(t_0)$ on $t\in[t_0-\varepsilon,t_0+\varepsilon]$.  If $f$ is bounded above by a constant and is continuous in $t$ and Lipschitz continuous in $y$  then, according to the theorem, for $\varepsilon>0$, there exists a unique solution $y(t)$ on $[t_0-\varepsilon,t_0+\varepsilon]$. This solution can be extended until either $y$ becomes unbounded or $y$ is no longer a solution.
The conditions of the theorem are satisfied because $\hat K(\Phi)$ and $\hat S$ are quadratic in $\Phi$, the RHS is Lipschitz continuous in $\Phi$ in the domain where $\hat K$ has only one zero eigenvalue, and the continuity of $\hat K$ and $\hat S$ in time follows immediately from the continuity of $\Phi$.

A nice connection of \eqref{eq:2a} to master equations in probabilistic processes can be drawn.  In  \eqref{eq:2a}, $\hat K$ has the form of a master equation for a probability distribution $P$,
\begin{eqnarray}
\partial_t P_n(t)&=&\sum_{n'} w_{nn'}P_{n'}(t)-w_{n'n}P_n(t)\\
&=&\sum_{n'}\left(w_{nn'}-\delta_{nn'}\sum_mw_{mn} \right)P_{n'}\phantom{spc}
\end{eqnarray}
with 
\begin{equation}
	w_{nn'}=-T_{nn'}\bra{\Phi(t)}( \hat a_n^\dag \hat a_{n'} +\hat a_{n'}^\dag\hat  a_n)\ket{\Phi(t)}.
	\label{eq:W}
\end{equation}
The key difference is that the entries of $K$ are not strictly positive ($\bra{\Phi(t)}\hat a_i^\dag \hat a_j\ket{\Phi(t)}$ can be positive or negative).  Since $K$ is Hermitian and its null space contains the uniform state, if all transition coefficients  were positive, then $K$ would satisfy detailed balance.

\section{Lemmas}\label{C}
\begin{lem}\label{lem:unitary} For two time-dependent Hamiltonians $H(t)=H_0+V(t)$ and $\tilde H(t)=H_0+\tilde V(t)$, the error in the evolution from $t_0$ to $t_1$ is bounded as
	\begin{equation}
		\|\Delta U(t_1,t_0)\|_2
		\leq (t_1-t_0)\;
		\max_{t_0\leq s \leq t_1}\left|V(s)-\tilde{V(s)}\right|_{\infty}
	\label{eq:err3}
	\end{equation}
\end{lem}
\begin{proof}
%
\begin{eqnarray*}
	U(t_1,t_0)-\tilde{U}({t_1,t_0}) &=& \tilde{U}(t_1,t_0) \left(\tilde{U}^\dag(t_1,t_0) U(t_1,t_0)-1\right)\\
  &=& \tilde U(t_1,t_0) \left(\int_{t_0}^{t_1}\frac d{ds} (\tilde
U^\dag({s,t_0}) U({s,t_0}))ds\right)\\
  &=& -i \tilde U(t_1,t_0) \left(\int_{t_0}^{t_1}\tilde
U^\dag({s,t_0}) (H(s)-\tilde H(s)) U(s,t_0)ds\right) \\
  &=& -i\int_{t_0}^{t_1} \tilde U({t_1,t_0})\tilde U(t_0,s)
(V(s)-\tilde V(s)) U(s,t_0)ds\\
  &=& -i\int_{t_0}^{t_1} \tilde U({t_1,s}) (V(s)-\tilde V(s)) U(s,t_0)ds
\end{eqnarray*}

Using sub-additivity and the unitary invariance of the operator norm
\begin{eqnarray*}
	\|U(t_1,t_0)-\tilde{U}({t_1,t_0})\|_2&\leq&(t_1-t_0)\max_{t_0\leq s\leq t_1}  \|V(s)-\tilde{V}(s)\|_2 
	\end{eqnarray*}
To obtain the statement in \eqref{eq:err3}, recall that for a diagonal matrix, the induced matrix 2-norm is the infinity norm of the corresponding vector of diagonal elements.  Noting that $V$ is diagonal gives $\|V\|_2=|V |_{\infty}$ to complete the proof.
\end{proof}

\begin{lem}\label{lem:linalg}  When we approximate the solution $x$ of $Ax=b$ from the solution, $\tilde x$, of $\tilde A \tilde x=\tilde b$, under the assumption that both $A$ and $\tilde A$ are invertible, the error in $x$ is bounded by
	\begin{eqnarray}
|\Delta x|&\leq& \alpha ( |\Delta b|+\|\Delta A\|\; |x|)
		\label{eq:key}
	\end{eqnarray}
where the vector and matrix norms are compatible (i.e. $|Mb|\leq\|M\||b|$).  
\end{lem}
\begin{proof}
	Define $\Delta x=x-\tilde x$ and similarly for $\Delta A$ and $\Delta b$.
	\begin{eqnarray*}
		|x-\tilde x|&=&  |A^{-1}b-A^{-1}\tilde b+A^{-1}\tilde b-\tilde A^{-1} \tilde b|\\
	 		&\leq& |A^{-1}\Delta b|+|(A^{-1}-\tilde A^{-1})\tilde b|\\
			&=&    |A^{-1}\Delta b|+|( A^{-1} \tilde A-\mathbf{1})\tilde A^{-1}\tilde b|\\
			&=&    |A^{-1}\Delta b|+| A^{-1}(\tilde A- A)\tilde x|\\
			&\leq& \|A^{-1}\|\;|\Delta b|+\|\tilde A^{-1}\|\|\tilde A- A\|\;|x|\\
		|\Delta x|	&\leq&\alpha\left( |\Delta b|+\|\Delta A\|\;|x|\right)
	\end{eqnarray*}
	Here, $\alpha=\max\{ \|A^{-1}\|,\|\tilde A^{-1}\|\}$. 
\end{proof}

\begin{lem}\label{lem:ddot}
Suppose density is measured with maximum error $|\Delta n|_\infty<\delta_n$ and the fourth derivative in time is bounded as $\max| \delta_t^4\Delta n|_\infty<c_4$, we have that
\begin{equation}
	|\Delta \partial_t^2 n|_\infty\leq \sqrt{2 c_4\delta_n}
	\label{eq:err5}
\end{equation}
\end{lem}
\begin{proof}
	We utilize the three point stencil to estimate the second derivative by Taylor expanding to third order
\begin{eqnarray*}
	f(t\pm h)&=& f(t)\pm \partial_t f (t) h +\frac12 \partial_t^2 f (t) h^2 +\pm\frac16  \partial_t^3 f (t) h^3+ R_3(t\pm h)\\
	R_3(t\pm h)&=& \frac{f^{(4)}(\xi)}{4!}h^4,\quad \textrm{for some }\xi \in [t,t\pm h] \\
	\partial_t^2 f(t)&=& \frac{f(t+h)-2f(t)+f(t-h)}{h^2}+\frac{R_3(t-h)+R_3(t+h)}{h^2}\\
	\left| \partial_t^2 f(t)-\partial_t^2 f^{3pt}\right|&\leq&\frac{f^{(4)}(\xi_1)+f^{(4)}(\xi_2)}{4!} h^2\leq\frac{ c_4h^2}{12}
\end{eqnarray*}
where $c_4$ is a bound for the fourth derivative of the function $f$.
\vskip1em
If $\delta_n$ is the maximum absolute difference between any component of the given density and the true density ($\infty$-norm of the difference) then from the triangle inequality,
\begin{eqnarray*}
|\partial_t^2 n(t)-\partial_t^2\tilde{{n}}(t)|_\infty&\leq&|\partial_t^2 n(t) -\partial_t^2 n(t)^{3pt}|_\infty+|\partial_t^2 n(t)^{3pt} -{\partial_t^2\tilde { n}}(t)|_\infty\\
&\leq&\frac{c_4}{12}h^2+\left|\frac{[n({t-h})-\tilde n({t-h})]-2[n({t})-\tilde n_(t)]+[n(t+h)-\tilde n(t+h)] }{h^2}\right|_\infty\\
|\Delta \partial_t^2 n|_\infty&\leq &\frac{c_4 h^2}{12}+\frac{4\delta_n}{h^2}
\end{eqnarray*}
To get the best bound, select $h^2=\sqrt{48 \delta_N/c_4}$. Substituting this into the previous equation gives,
\begin{equation}
	|\Delta \partial_t^2 n|_\infty\leq 
	\left(\frac{\sqrt{48}}{12}+\frac4{\sqrt{48}}	\right)\sqrt{\delta_n c_4}< \sqrt{2}\sqrt{\delta_nc_4}
\end{equation}

\end{proof}

\end{document}